\documentclass[runningheads]{llncs}

\usepackage[T1]{fontenc}
\usepackage[utf8]{inputenc}
\usepackage{subcaption}
\usepackage{tikz}
\usetikzlibrary{arrows}
\usepackage{xspace}
\usepackage{amsfonts,amsmath,amssymb}
\usepackage[colorlinks=true, citecolor=blue]{hyperref}

\usepackage[capitalize,nameinlink]{cleveref}
\crefname{algocf}{alg.}{algs.}
\Crefname{algocf}{Algorithm}{Algorithms}

\usepackage{multirow}

\usepackage[linesnumbered,algoruled, lined, noend]{algorithm2e}

\usepackage[shortlabels]{enumitem}
\setitemize{noitemsep,topsep=0pt,parsep=0pt,partopsep=0pt}

\newcommand{\setdisj}{\textsf{Set Disjointness }}

\def \ba     {{\bf a}}
\def \bx     {{\bf x}}
\def \by     {{\bf y}}

\providecommand{\mS}[0]{\mathcal{S}}
\providecommand{\domain}[0]{\mathbf{dom}}

\providecommand{\edges}[0]{\mathcal{E}}
\providecommand{\nodes}[0]{\mathcal{V}}
\providecommand{\vars}[1]{\textsf{vars}(#1)}
\providecommand{\bx}[0]{\mathbf{x}}
\providecommand{\by}[0]{\mathbf{y}}

\providecommand{\bu}[0]{\mathbf{u}}
\newcommand{\introparagraph}[1]{\noindent {\bf \em #1.}}
\providecommand{\htree}[0]{\mathcal{T}}
\providecommand{\bag}[0]{\mathcal{B}}
\providecommand{\fhw}[1]{\mathsf{fhw}(#1)}
\providecommand{\bound}[0]{\mathsf{b}}
\providecommand{\free}[0]{\mathsf{f}}
\providecommand{\slack}[0]{\alpha}
\providecommand{\mB}[0]{\mathcal{B}}
\providecommand{\prq}[2]{{#1}\mid{#2}}
\providecommand{\tree}[0]{\mathtt{T}}
\providecommand{\manc}[0]{\mathsf{anc}}

\providecommand{\mL}[0]{\mathcal{L}}
\usepackage{cite}
\usepackage[appendix=strip]{apxproof}
\newtheoremrep{theorem}{Theorem}
\usepackage{tikz}

\begin{document}
\title{General Space-Time Tradeoffs via Relational Queries}

\author{
    Shaleen Deep\inst{1}
      \and Xiao Hu\inst{2}
      \and Paraschos Koutris\inst{3}
}
\authorrunning{Deep et al.}
% First names are abbreviated in the running head.
% If there are more than two authors, 'et al.' is used.
%
\institute{Microsoft GSL, USA \and University of Waterloo, Canada \and University of Wisconsin-Madison, USA \\
\email{shaleen.deep@microsoft.com},
\email{xiaohu@uwaterloo.ca},
\email{paris@cs.wisc.edu}}
\maketitle              % typeset the header of the contribution
\begin{abstract}
    In this paper, we investigate space-time tradeoffs for answering Boolean conjunctive queries. The goal is to create a data structure in an initial preprocessing phase and use it for answering (multiple) queries. Previous work has developed data structures that trade off space usage for answering time and has proved conditional space lower bounds for queries of practical interest such as the path and triangle query. However, most of these results cater to only those queries, lack a comprehensive framework, and are not generalizable. The isolated treatment of these queries also fails to utilize the connections with extensive research on related problems within the database community. The key insight in this work is to exploit the formalism of relational algebra by casting the problems as answering join queries over a relational database. Using the notion of boolean {\em adorned queries} and {\em access patterns}, we propose a unified framework that captures several widely studied algorithmic problems. Our main contribution is three-fold. First, we present an algorithm that recovers existing space-time tradeoffs for several problems. The algorithm is based on an application of the {\em join size bound} to capture the space usage of our data structure. We combine our data structure with {\em query decomposition} techniques to further improve the tradeoffs and show that it is readily extensible to queries with negation. Second, we falsify two proposed conjectures in the existing literature
    related to the space-time lower bound for path queries and triangle detection for which we show unexpectedly better algorithms. This result opens a new avenue for improving several algorithmic results that have so far been assumed to be (conditionally) optimal. Finally, we prove new conditional space-time lower bounds for star and path queries. 
\end{abstract}
\section{Introduction} \label{sec:intro}

Recent work has made remarkable progress in developing data structures and algorithms for answering set intersection problems~\cite{goldstein2017conditional}, reachability oracles and directed reachability~\cite{agarwal2011approximate,agarwal2014space,cohen2010hardness}, histogram indexing~\cite{chan2015clustered,kociumaka2013efficient}, and problems related to document retrieval~\cite{afshani2016data,larsen2015hardness}. This class of problems splits an algorithmic task into two phases: the {\em preprocessing phase}, which computes a space-efficient data structure, and the {\em answering phase}, which uses the data structure to answer the requests to minimize the answering time. A fundamental algorithmic question related to these problems is the tradeoff between the space $S$ necessary for data structures and the answering time $T$ for requests. 

For example, consider the $2$-\setdisj problem: given a universe of elements $U$ and a collection of $m$ sets $C_1, \dots, C_m \subseteq U$, we want to create a data structure such that for any pair of integers $1 \leq i,j \leq m$, we can efficiently decide whether $C_i \cap C_j$ is empty or not. Previous work~\cite{cohen2010hardness,goldstein2017conditional} has shown that the space-time tradeoff for $2$-\setdisj is captured by the equation $S \cdot T^2 = N^2$, where $N$ is the total size of all sets. The data structure obtained is conjectured to be optimal~\cite{goldstein2017conditional}, and its optimality was used to develop conditional lower bounds for other problems, such as approximate distance oracles~\cite{agarwal2011approximate,agarwal2014space}. Similar tradeoffs have been independently established for other data structure problems as well. In the $k$-\textsf{Reachability} problem~\cite{goldstein2017conditional,Cohen2010} we are given as an input a directed  graph $G = (V,E)$, an arbitrary pair of  vertices $u, v$, and the goal is to decide whether there exists a path of length $k$ between $u$ and $v$. In the \textsf{edge triangle detection} problem~\cite{goldstein2017conditional}, we are given an input undirected  graph $G = (V,E)$, and the goal is to develop a data structure that takes space $S$ and can answer in time $T$ whether a given edge $e \in E$ participates in a triangle or not. Each of these problems has been studied in isolation and, as a result, the algorithmic solutions are not generalizable. %This has also led to no progress in improving the state-of-the-art results.% and as a result, their optimality has been conjectured. 

In this paper, we cast many of the above problems into answering {\em Conjunctive Queries (CQs)} over a relational database. CQs are a powerful class of relational queries with widespread applications in data analytics and graph exploration~\cite{graphgen2015,graphgen2017,deep2018compressed}. For example, by using the relation $R(x,y)$ to encode that element $x$ belongs to set $y$, $2$-\setdisj can be captured by  the following CQ: $\varphi (y_1, y_2) = R(x,y_1) \land R(x,y_2)$. The insight of casting data structure problems into CQs over a database allows for a unified treatment for developing algorithms within the same framework. In particular, we can leverage the techniques developed by the data management community through a long line of research on efficient join evaluation~\cite{yannakakis1981algorithms,skewstrikesback,ngo2012worst}, including worst-case optimal join algorithms~\cite{ngo2012worst} and tree decompositions~\cite{gottlob2014treewidth,robertson1986graph}. Building upon these techniques, we achieve the following:
\begin{itemize}
\item We obtain in a simple way general space-time tradeoffs for any Boolean CQ (a Boolean CQ is one that outputs only true or false). As a consequence, we recover state-of-the-art tradeoffs for several existing problems (e.g., $2$-\setdisj as well as its generalization $k$-\setdisj and $k$-\textsf{Reachability}) as special cases of the general tradeoff. We can even obtain improved tradeoffs for some specific problems, such as edge triangles detection, thus falsifying existing conjectures. This also gives us a way to construct data structures for any new problem that can be cast as a Boolean CQ (e.g., finding any subgraph pattern in a graph).
\item Space-time tradeoffs for enumerating (non-Boolean) query results under static and dynamic settings have been a subject of previous work~\cite{abo2020decision,greco2013structural,deep2018compressed,olteanu2016factorized,kara19,kara2019counting}.  The space-time tradeoffs from~\cite{deep2018compressed} can be applied to the setting of this paper by stopping the enumeration after the first result is observed. We improve upon this result by $(i)$ showing a much simpler data structure construction and proofs, and $(ii)$ shaving off a polylogarithmic factor from the tradeoff.

\end{itemize}

\noindent We next summarize our three main technical contributions.

\begin{enumerate}
%\item We propose a unified framework that captures several widely-studied data structure problems.  We then show how this framework captures the $2$-\setdisj and $k$-\textsf{Reachability} problems. 

\item We propose a unified framework that captures several widely-studied data structure problems. More specifically, we use the formalism of CQs and the notion of {\em Boolean adorned queries}, where the values of some variables in the query are fixed by the user (denoted as an {\em access pattern}), and aim to evaluate the Boolean query. We then show how this framework captures the $2$-\setdisj and $k$-\textsf{Reachability} problems. Our first main result (\autoref{thm:main1}) is an algorithm that builds a data structure to answer any Boolean CQ under a specific access pattern. We show how to recover existing and new tradeoffs using this general framework. The first main result may sometimes lead to suboptimal tradeoffs since it does not take into account the structural properties of the query. Our second main result (\autoref{thm:main2}) combines tree decompositions of the query structure with access patterns to improve space efficiency. We then show how this algorithm can handle Boolean CQs with negation.

\item We explicitly improve the best-known space-time tradeoff for the $k$-\textsf{Reachability} problem for  $k \geq 4$. For any $k \geq 2$, the tradeoff of $S \cdot T^{2/(k-1)} = O(|E|^2)$ was conjectured to be optimal by~\cite{goldstein2017conditional}, where $|E|$ is the number of edges in the graph, and was used to conditionally prove other lower bounds on space-time tradeoffs. We show that for a regime of answer time $T$, it can be improved to $S \cdot T^{2/(k-2)} = O(|E|^2)$, thus breaking the conjecture. %making the conjecture false. 
	To the best of our knowledge, this is the first non-trivial improvement for the $k$-\textsf{Reachability} problem. We also refute a lower bound conjecture for the edge triangles detection problem established by~\cite{goldstein2017conditional} that appeared at WADS'17.

\item Our third main contribution applies our framework to CQs with negation. This allows us to construct space-time tradeoffs for tasks such as detecting open triangles in a  graph. We also show a reduction between lower bounds for the problem of $k$-\setdisj for $k \ge 2$, which generalizes the $2$-\setdisj to computing the intersection between $k$ given sets.

%We then show how this algorithm can handle Boolean CQs with negation.
	
%\item \introparagraph{Improved Algorithms} In addition to the main result above, we explicitly improve the best-known space-time tradeoff for the $k$-\textsf{Reachability} problem for  $k \geq 4$. For any $k \geq 2$, the tradeoff of $S \cdot T^{2/(k-1)} = O(|E|^2)$ was conjectured to be optimal by~\cite{goldstein2017conditional}, where $|E|$ is the number of edges in the graph, and was used to conditionally prove other lower bounds on space-time tradeoffs. We show that for a regime of answer time $T$, it can be improved to $S \cdot T^{2/(k-2)} = O(|E|^2)$, thus breaking the conjecture. %making the conjecture false. To the best of our knowledge, this is the first non-trivial improvement for the $k$-\textsf{Reachability} problem. 

 \end{enumerate}	

%\smallskip
%\introparagraph{Organization} We introduce the basic terminology and problem definition in~\Cref{sec:framework} and~\Cref{sec:framework2}. We presents our main results for Boolean adorned queries in \Cref{sec:main1} to \Cref{sec:negation} and our improved result for $k$-\textsf{Reachability} queries in ~\Cref{sec:path}. We discuss the lower bounds and related work in~\Cref{sec:lowerbound} and~\Cref{sec:related}, and finally conclude in~\Cref{sec:conclusion} with promising future research directions and open problems. 

\section{Notation and Preliminaries} \label{sec:framework}

\renewcommand{\mS}{\mathbf{s}}

\introparagraph{Data Model} A {\em schema} is defined as a collection of relation names, where each relation name $R$ is associated with an arity $n$. Assuming a (countably infinite) domain $\domain$, a tuple $t$ of relation $R$ is an element of $\domain^n$. An instance of relation $R$ with arity $n$ is a finite set of tuples of $R$; the size of the instance will be denoted as $|R|$.  An input database $D$ is a set of relation instances over the schema. The size of the database $|D|$ is the sum of sizes of all its instances.

%A schema $\bx = (x_1, \dots , x_n)$ is a non-empty ordered set of distinct variables. Each variable $x_i$ has a discrete domain $\domain(x_i)$. A tuple $t$ over schema $\bx$ is an element from $\domain(\bx) = \domain(x_1) \times \dots \times \domain(x_n)$. A relation $R$ over schema $\bx$ (denoted $R(\bx)$) is a function $R : \domain(\bx) \rightarrow \mathbb{Z}$ such that the multiplicity $R(t)$ is non-zero for finitely many $t$. A tuple $t$ exists in $R$, denoted by $t \in R$, if $R(t) > 0$. The size of relation $R$, denoted as $|R|$,  is the size of set $\setof{t}{t \in R}$. A database $D$ is a set of relations and the size of the database $|D|$ is the sum of sizes of all its relations. Given a tuple $t$ over schema $\bx$ and a set of variables $\mS \subseteq \bx$, $t[\mS]$ denotes the restriction of $t$ to $\mS$ and the values of $t[\mS]$ follows the same variable ordering as $\mS$. We also define the {\em selection} operator $\sigma_{\mS = t} (R) = \setof{u \in R}{ u[\mS] = t}$ and {\em projection} operator $\pi_{\mS} (R) = \setof{u[\mS]}{ u \in R}$.

\smallskip

\introparagraph{Conjunctive Queries} A {\em Conjunctive Query} (CQ) is an expression of the form $\varphi(\by) = R_1(\bx_1) \land R_2(\bx_2) \land \ldots \land R_n(\bx_n).$
The expressions $\varphi(\by), R_1(\bx_1), R_2(\bx_2), \ldots, R_n(\bx_n)$ are called {\em atoms}. The atom $\varphi(\by)$ is the {\em head} of the query, while the atoms $R_i(\bx_i)$ form the {\em body}. Here, $\by, \bx_1, \dots, \bx_n$ are vectors where each position is a variable (typically denoted as $x,y,z,\dots$) or a constant from $\domain$ (typically denoted  $a,b,c,\dots$). Each $\bx_i$ must match the arity of the relation $R_i$, and the variables in $\by$ must occur in the body of the query. We use $\vars{\varphi}$ to denote the set of all variables occurring in $\varphi$, and $\vars{R_i}$ to denote the set of variables in atom $R_i(\bx_i)$.
%The {\em existential} quantified variables $\bz$ is the set of variables $\vars{\varphi} \setminus \by$.  Throughout the paper, we will omit the existential quantified part whenever $\by$ and $\bx_i$ are mentioned in the query. 
A CQ is {\em full} if every variable in the body appears also in the head, and {\em Boolean} if the head contains no variables. Given variables $x_1, \dots, x_k$ from $\vars{\varphi}$ and constants $a_1, \dots, a_k$ from $\domain$, we define $\varphi[a_1 /x_1, \dots, a_k / x_k]$ to be the CQ where every occurrence of a variable $x_i$, $i=1, \dots, k$, is replaced by the constant $a_i$. Given an input database $D$ and a CQ $\varphi$, we define the query result $\varphi(D)$ as follows. A {\em valuation} $v$ is a mapping from $\domain \cup \vars{\varphi}$ to $\domain$ such that $v(a)=a$ whenever $a$ is a constant. Then, $\varphi(D)$ is the set of all tuples $t$ such that there exists a valuation $v$ for which $t = v(\by)$ and for every atom $R_i(\bx_i)$, we have $R_i(v(\bx_i)) \in D$.\footnote{Here we extend the valuation to mean $v((a_1, \dots, a_n)) = (v(a_1), \dots, v(a_n))$.}

%A CQ $\varphi$ can be represented as a {\em hypergraph}  $\mathcal{H}_\varphi = (\nodes_\varphi, \edges_\varphi)$, where $\nodes_\varphi = \vars{\varphi}$, and each hyperedge $F \in \edges_\varphi$ corresponds to a relation $R_F$ with variables $F$. 

\begin{example}
Suppose that we have a directed graph $G$ that is represented through a binary relation $R(x,y)$: this means that there exists an edge from node $x$ to node $y$. We can compute the pairs of nodes that are connected by a directed path of length $k$ using the following CQ, which we call a {\em path query}: $
P_{k}(x_1, x_{k+1}) = R(x_1, x_2) \land R(x_2, x_3) \land \dots \land R(x_k, x_{k+1}).$
%The hypergraph of this query has $\nodes = \{x_1, \dots, x_{k+1}\}$, and $\edges = \{ \{x_1, x_2\}, \{x_2, x_3\}, \dots, \{x_k, x_{k+1}\}\}$.
\end{example}

\smallskip

%\phantomsection
%\label{the_label}
\introparagraph{Output Size Bounds}
Let $\varphi(\by) = R_1(\bx_1) \land R_2(\bx_2) \land \ldots \land R_n(\bx_n)$ be a CQ.
A weight assignment $\bu = (u_i)_{i=1, \dots, n}$ 
is called a {\em fractional edge cover} of $S \subseteq \vars{\varphi}$ if 
$(i)$ for every atom $R_i$, $u_i \geq 0$  and $(ii)$ for every
$x \in S, \sum_{i:x \in \vars{R_i}} u_i \geq 1$. 
The {\em fractional edge cover number} of $S$, denoted by
$\rho^*(S)$ is the minimum of
$\sum_{i=1}^n u_i$ over all fractional edge covers of $S$. Whenever $S = \vars{\varphi}$, we call this a fractional edge cover of $\varphi$ and simply use $\rho^*$. In a celebrated result, Atserias, Grohe and Marx~\cite{AGM} proved that for every fractional edge cover $\bu$ of $\varphi$,  the size of the output is bounded by the {\em AGM inequality}:
$
|\varphi(D)| \leq \prod_{i=1}^n |R_i|^{u_i}
$.
The above bound is constructive~\cite{skewstrikesback,ngo2012worst}: there exists an algorithm that computes the result $\varphi(D)$ in  $O(\prod_{i} |R_i|^{u_i})$ time for every fractional edge cover $\bu$.

\smallskip

\introparagraph{Tree Decompositions}
Let $\varphi(\by) = R_1(\bx_1) \land R_2(\bx_2) \land \ldots \land R_n(\bx_n)$ be a CQ.
A {\em tree decomposition} of $\varphi$ is  a tuple 
$(\htree, (\bag_t)_{t \in V(\htree)})$ where $\htree$ is a tree, and 
every $\bag_t$ is a subset of $\vars{\varphi}$, called the {\em bag} of $t$, such that 
	\begin{itemize}
		\item  For every atom $R_i$,  the set $\vars{R_i}$ is contained in some bag; and
		\item  For each variable $x \in \vars{\varphi}$, the set of nodes $\{t \mid x \in \bag_t\}$ form a connected subtree of $\htree$.
	\end{itemize}

The {\em fractional hypertree width} of a decomposition is 
defined as $\max_{t \in V(\htree)} \rho^*(\bag_t)$, where
$\rho^*(\bag_t)$ is the minimum fractional edge cover of the vertices in $\bag_t$.
The  fractional hypertree width of a query $\varphi$, denoted $\fhw{\varphi}$, is the minimum 
fractional hypertree width among all tree decompositions.
We say that a query is {\em acyclic} if $\fhw{\varphi}=1$.

\smallskip

\introparagraph{Computational Model}
To measure the running time of our algorithms, we will use the uniform-cost RAM  model~\cite{hopcroft1975design}, where data values and pointers to databases are of constant size. Throughout the paper, all complexity results are with respect to data complexity, where the query is assumed fixed. 

%Each relation $R$ over schema $\bx$ is implemented via a data structure that stores all entries $t \in R$ in $O(|R|)$ space, which supports look-up, insertion, and deletion entries in $O(1)$ time. For a schema $\mS \subseteq \bx$, we use an index structure that for some $t$ defined over $\mS$ can (i) check if $t \in \pi_{\mS} (R)$; and $(ii)$ return $|\sigma_{\mS = t} (R)|$ in constant time.

%%%%%%%%%%%%%%%%%%%%%%%%%%%%%%%%%%%
\section{Framework} \label{sec:framework2}

%In this section, we discuss the concept of adorned queries and present our framework.

\subsection{Adorned Queries}

In order to model different access patterns, we will use the concept of {\em adorned queries} introduced by~\cite{ullman1986approach}. Let $\varphi(x_1, \dots, x_k)$ be the head of a CQ $\varphi$. In an adorned query, each variable in the head is associated with a binding type, which can be either {\em bound} ($\bound$) or {\em free} ($\free$). We denote this as $\varphi^\eta$, where $\eta \in \{\bound,\free\}^k$ is called the {\em access pattern}. The access pattern tells us for which variables the user must provide a value as input. Concretely, let $x_1,x_2, \dots, x_\ell$ be the bound variables. An {\em access request} is sequence of constants $a_1, \dots, a_\ell$, and it asks to return the result of the query $\varphi^\eta[a_1 /x_1, \dots, a_\ell / x_\ell]$ on the input database. We next demonstrate how to capture several data structure problems  in this way.

\begin{example}[Set Disjointness and Set Intersection]
In the set disjointness problem, we are given $m$ sets $S_1, \dots, S_m$ drawn from the same universe $U$. Let $N = \sum_{i=1}^m |S_i|$ be the total size of input sets. Each access request is a pair of indexes $(i,j), 1 \leq i,j, \leq m$, for which we need to decide whether $S_i \cap S_j$ is empty or not.
To cast this problem as an adorned query, we encode the family of sets as a binary relation $R(x,y)$, such that element $x$ belongs to set $y$. Note that the relation will have size $N$. Then, the set disjointness problem corresponds to:
$\varphi^{\bound\bound}(y,z) =  R(x,y) \land R(x,z).$
An access request in this case specifies two sets $y=S_i, z=S_j$, and issues the (Boolean) query $\varphi(S_i, S_j) =  R(x,S_i) \land R(x, S_j)$. In the related set intersection problem, given a pair of indexes $(i,j)$ for $1 \leq i,j, \leq m$, we instead want to enumerate the elements in the intersection $S_i \cap S_j$, which can be captured by the following adorned query:
$ \varphi^{\bound\bound\free}(y,z,x) = R(x,y) \land R(x,z)$.
\end{example}

\begin{example}[$k$-Set Disjointness]
The $k$-set disjointness problem is a generalization of 2-set disjointness problem, where each request asks whether the intersection between $k$ sets is empty or not. Again, we can cast this problem into the following adorned query:
$ \varphi^{\bound \dots \bound}(y_1, \dots, y_k) =  R(x,y_1) \land R(x,y_2) \land \dots \land R(x,y_k)$
\end{example}

\begin{example}[$k$-Reachability]
Given a direct graph $G$	, the $k$-reachability problem asks, given a pair vertices $(u,v)$, to check whether they are connected by a path of length $k$. Representing the graph as a binary relation $R(x,y)$ (which means that there is an edge from $x$ to $y$), we can model this problem through the following adorned query:
$
\varphi^{\bound\bound}(x_1, x_{k+1}) =  R(x_1, x_2) \land R(x_2, x_3) \land \dots \land R(x_k, x_{k+1})
$	
Observe that we can also check whether there is a path of length at most $k$ by combining the results of $k$ such queries (one for each length $1, \dots, k$).
\end{example}

\begin{example}[Edge Triangles Detection]
	Given a graph $G = (V, E)$, this problem asks, given an edge $(u,v)$ as the request, whether $(u,v)$ participates in a triangle or not. This task can be expressed as the following adorned query $\varphi^{\bound \bound}_\triangle (x,z) =  R(x,y) \land R(y,z) \land R(x,z)$
	In the reporting version, the goal is to enumerate all triangles participated by edge $(x,z)$, which can also be expressed by the following adorned query
	$
	\varphi^{\bound \bound \free}_\triangle (x,z,y) = R(x,y) \land R(y,z) \land R(x,z)
	$.
\end{example}	

We say that an adorned query is {\em Boolean} if every head variable is bound. In this case, the answer for every access request is also Boolean, i.e., true or false.

\subsection{Problem Statement}

Given an adorned query $\varphi^\eta$ and an input database $D$, our goal is to construct a data structure, such that we can answer any access request that conforms to the access pattern $\eta$ as fast as possible. In other words, an algorithm can be split into two phases:
\begin{itemize}
	\item {\bf Preprocessing phase:} we compute a data structure using space {\em $S$}. %We will often need to calculate the time to construct the data structure, denoted $P$.
 \smallskip
	\item {\bf Answering phase:} given an access request, we compute the answer using the data structure built in the preprocessing phase, within time $T$. %We denote by $T$ the time needed to compute any access request. 
\end{itemize}

In this work, our goal is to study the relationship between the space of the data structure $S$ and the answering time $T$ for a given adorned query $\varphi^\eta$. We will focus on Boolean adorned queries, where the output is just true or false.
\section{Space-Time Tradeoffs via Worst-case Optimal Algorithms}
\label{sec:main1}

Let $\varphi^{\eta}$ be an adorned query and $\nodes_\bound$ denote its bound variables. For any fractional edge cover $\bu$, we define the {\em slack} of $\bu$~\cite{deep2018compressed}  as:
$$ \slack(\bu) := \min_{x \in \vars{\varphi} \setminus \nodes_\bound} \left( \sum_{i: x \in \vars{R_i}} u_i \right). $$
In other words, the slack is the maximum factor by which we can scale down the fractional cover $\bu$ so that it remains a valid edge cover of the non-bound variables in the query\footnote{We will omit the parameter $\bu$ from the notation of $\alpha$ whenever it is clear from the context.}. Hence $\{u_i/\slack(\bu)\}_{i}$ is a fractional edge cover of the nodes in $ \vars{\varphi} \setminus \nodes_\bound$. We always have $\slack(\bu) \geq 1$.

\begin{example}
Consider $\varphi^{\bound \dots \bound}(y_1, \dots, y_k) =  R_1(x,y_1)\land R_2(x,y_2) \land \dots R_k(x, y_k)$ with the optimal fractional edge cover $\bu$, where $u_i=1$ for $i \in \{1, \dots, k\}$. The slack is $\slack(\bu) = k$, since the fractional edge cover $\hat{\bu}$, where $\hat{u}_i = u_i/k = 1/k$ covers the only non-bound variable $x$.	
\end{example}

\begin{theoremrep}\label{thm:main1}
Let $\varphi^{\eta}$ be a Boolean adorned query. Let $\bu$ be any fractional edge cover of $\varphi$. 
Then, for any input database $D$, we can construct a data structure that answers any access request in time $O(T)$ and takes space
$$S = {O}\left(|D| + \prod_{i=1}^n |R_i|^{u_i} / T^\slack  \right)$$
\end{theoremrep}

 \begin{proof}
To simplify the presentation, we will assume that $\varphi$ contains no constants in the body or repeated variables in the same atom, but our approach can be easily extended to cover these cases.
Let $\nodes_\bound = \{x_1, \dots, x_k\}$ be the set of bound variables. Recall that an access request $\ba = (a_1, \dots, a_k)$ corresponds to the query $\varphi[a_1/x_1, \dots, a_k/x_k]$. 	
For every atom $R_i(\bx_i)$ in the query, define $R_i(\ba) = \sigma_{x_j =a_j \mid x_j \in \nodes_\bound \cap \vars{R_i}}(R_i)$. Here, $\sigma_\psi$ is a selection operator that keeps the tuples from $R_i$ that satisfy the condition $\psi$. We say that an access request $\ba$ is {\em valid} if for every atom, $R_i(\ba) \neq \emptyset$.

%Define the hypergraph $\hgraph_\bound = (\nodes_\bound, \edges_\bound)$, where $\edges_\bound = \{F \cap \nodes_\bound \mid F \in \edges\}$, and let $\varphi_\bound$ be the query that corresponds to $\hgraph_\bound$ with head variables $\nodes_\bound$.  We say that an access request $\ba$ is {\em valid} if $\ba \in \varphi_\bound(D)$. 

If $\slack$ is the slack for the fractional edge cover $\bu$, define $\hat{u}_i = u_i / \slack$. Note that $\hat{\bu} = \{\hat{u}_i\}_{i}$ is a fractional edge cover for the query $\varphi[a_1/x_1, \dots, a_k/x_k]$: indeed, it is sufficient to cover only the non-bound variables, since all bound variables are replaced by constants in the query. Hence, using a worst-case optimal join algorithm, we can compute the access request $\varphi[a_1/x_1, \dots, a_k/x_k]$ with running time
	$$ T(\ba) = \prod_{i=1}^n |R_i(\ba)|^{u_i/\slack}.$$
	
The time required is a direct application of the upper bound of worst-case optimal join algorithm running time as described in~\Cref{sec:framework}. We can now describe the data structure we build. First, for every atom $R_i(\bx_i)$ we build a standard hash index that returns in constant time the subset $R_i(\ba)$ for every access request $\ba$. Second, we create a hash index $\mathcal{K}$. Let $J$ be the set of valid access requests such that $T(\ba) > T$. For every $\ba \in J$, we add to the hash index the key-value entry $(\ba, \varphi[a_1/x_1, \dots, a_k/x_k](D))$. In other words, the value is the (boolean) answer to the access request $\ba$.

%We can construct hash indexes of linear size $O(|D|)$ during the preprocessing phase so that we can check whether any access request is valid in constant time $O(1)$. 

	We claim that the answer time using the above data structure is at most $O(T)$. Indeed, we first check whether $\ba$ is valid, which we can do in constant time. If it is not valid, we simply output no. If it is valid, we probe the hash index $\mathcal{K}$. If $\ba$ exists in the hash index, we obtain the answer in time $O(1)$ by reading the value of the corresponding entry. Otherwise, we know that $T(\ba) < T$ and hence we can compute the answer to the access request in time $O(T)$ using a worst-case optimal join. 
	
It remains to bound the size of the data structure we constructed during the preprocessing phase. Since the size is $O(|D|+|J|)$, we will bound the size of $J$. Indeed, we have:
	\allowdisplaybreaks
	\begin{align*}
	\hspace{7em} T \cdot |J|  \leq \sum_{\ba \in J} T(\ba) 
	& = \sum_{\ba \in J} \prod_{i} |R_i(\ba)|^{u_i/\slack} \\
	& = \sum_{\ba \in J} 1^{1-1/\slack} \cdot \left( \prod_{i} |R_i(\ba)|^{u_i} \right)^{1/\slack} \\
	& \leq \left(\sum_{\ba \in J} 1 \right)^{1-1/\slack} \cdot \left( \sum_{\ba \in J} \prod_{i} |R_i(\ba)|^{u_i} 	\right)^{1/\slack} \\
	%%& \leq |J|^{1-1/\slack} \cdot \prod_{F \in \edges} \left( \sum_{\ba \in J}|R_F(\ba)|\right)^{u_F/\slack} \\
	& \leq |J|^{1-1/\slack}  \cdot \prod_{i} |R_i|^{u_i/\slack}
	\end{align*}
	
	The first inequality follows directly from the definition of the set $J$. The second inequality is H{\"o}lders inequality. The third inequality is an application of the query decomposition lemma from~\cite{skewstrikesback}.	By rearranging the terms, we can now obtain the desired bound.
\end{proof}
   
We should note that~\autoref{thm:main1} applies even when the relation sizes are different; this gives us sharper upper bounds compared to the case where  each relation is bounded by the total size of the input. Indeed, if using $|D|$ as an upper bound on each relation, we obtain a space requirement of $O(|D|^{\rho^*}/T^\slack)$ for  achieving answering time $O(T)$, where $\rho^*$ is the fractional edge cover number.
Since $\slack \geq 1$, this gives us at worst a linear tradeoff between space and time, i.e., $S \cdot T = O(|D|^{\rho^*})$. For cases where $\slack \ge 1$, we can obtain better tradeoffs. The full proofs for all results in this paper can be found in~\cite{deep2021space}.

\begin{example}
Continuing the example in this section $\varphi^{\bound \dots \bound}(y_1, \dots, y_k) =  R_1(x,y_1) \land R_2(x,y_2) \land \dots \land R_k(x, y_k)$. We obtain an improved tradeoff: $S \cdot T^{k} = O(|D|^k)$\footnote{For all results in this paper, $S$ includes the space requirement of the input as well. If we are interested in only the space requirement of the constructed data structure, then the $|D|$ term in the space requirement of~\autoref{thm:main1} can be removed.}. Note that this result matches the best-known space-time tradeoff for the $k$-\setdisj problem~\cite{goldstein2017conditional}. (Note that all atoms use the same relation symbol $R$, so $|R_i| = |D|$ for every $i=1, \dots, k$. )
\end{example}

\begin{example}[Edge Triangles Detection] For the Boolean version, it was shown in~\cite{goldstein2017conditional} that -- conditioned on the strong set disjointness conjecture -- any data structure that achieves answering time $T$ needs space  $S = \Omega(|E|^2/T^2)$. A matching upper bound can be constructed by using a fractional edge cover $\mathbf{u} = (1,1,0)$ with slack $\slack = 2$. Thus,~\autoref{thm:main1} can be applied to achieve answering time $T$ using space $S  = O(|E|^2/T^2)$. Careful inspection reveals that  a different fractional edge cover $\mathbf{u} = (1/2,1/2,1/2)$ with slack $\alpha = 1$,  achieves a better tradeoff. Thus,~\autoref{thm:main1} can be applied to obtain the following corollary.
\end{example}

\begin{corollary}
	For a graph $G=(V,E)$, there exists a data structure of size $S = O(|E|^{3/2}/T)$ that can answer the edge triangles detection problem in $O(T)$.
\end{corollary}

The data structure implied by~\autoref{thm:main1} is always better when $T \leq \sqrt{|E|}$\footnote{All answering times $T > \sqrt{|E|}$ are trivial to achieve using linear space by using the data structure for $T' = \sqrt{E}$ and holding the result back until time $T$ has passed.}, thus refuting the conditional lower bound in~\cite{goldstein2017conditional}. We should note that this does not imply that the strong set disjointness conjecture is false, as we have observed an error in the reduction used by~\cite{goldstein2017conditional}.

\begin{example}[Square Detection]
	Beyond triangles, we consider the edge square detection problem, which checks 
	%We consider other graph pattern detection, in addition to triangles. For example, the adorned query below checks 
	whether a given edge belongs in a square pattern in a graph $G = (V,E)$,
$ \varphi^{\bound\bound}_\square(x_1, x_2) =  R_1(x_1,x_2) \land R_2(x_2,x_3) \land R_3(x_3,x_4) \land R_4(x_4,x_1).$
Using the fractional edge cover $\bu = (1/2,1/2,1/2,1/2)$ with slack $\slack =1$,
%Considering the fractional edge cover that assigns a weight of $1/2$ to each hyperedge, with slack $\slack=1$, 
we obtain a tradeoff $S = O(|E|^2/T)$.
\end{example}

\section{Space-Time Tradeoffs via Tree Decompositions}
\label{sec:main2}

\autoref{thm:main1} does not always give us the optimal tradeoff. For the $k$-reachability problem with the adorned query $\varphi^{\bound \bound} (x_1, x_{k+1}) =  R_1(x_1, x_2) \land \dots \land R_k(x_k, x_{k+1})$, \autoref{thm:main1} gives a tradeoff $S \cdot T = |D|^{\lceil (k+1)/2\rceil}$, by taking the optimal fractional edge covering number $\rho^* = \lceil (k+1)/2\rceil$ and slack $\slack=1$, which is far from efficient. In this section, we will show how to leverage tree decompositions to further improve the space-time tradeoff in~\autoref{thm:main1}. 

Again, let $\varphi^{\eta}$ be an adorned query. Given a set of nodes $C \subseteq \nodes$,  a $C$-connex tree decomposition of $\varphi$ is a pair $(\htree,A)$, where $(i)$ $\htree$ is a tree decomposition of $\varphi$, and $(ii)$ $A$ is a connected subset of the tree nodes such that the union of their variables is exactly $C$. For our purposes, we choose $C = \nodes_\bound$. Given a $\nodes_\bound$-connex tree decomposition, we orient the tree from some node in $A$. We then define the bound variables for the bag $t$, $\nodes^t_\bound$ as the variables in $\bag_t$ that also appear in the bag of some ancestor of $t$. The free variables for the bag $t$ are the remaining variables in the bag, $\nodes^t_\free = \bag_t \setminus \nodes^t_\bound$.

\begin{example}
	Consider the $5$-path query $\varphi^{\bound \bound} (x_1, x_{6}) = R_1(x_1, x_2) \land \dots \land R_5(x_5, x_{6})$. Here, $x_1$ and $x_6$ are the bound variables.~\autoref{fig:cfhw} shows the unconstrained decomposition as well as the $C$-connex decomposition for $\varphi^{\bound \bound} (x_1, x_{6})$, where $C = \{x_1, x_6\}$. The root bag contains the bound variables $x_1, x_6$. Bag $\mB_{t_2}$ contains $x_1, x_6$ as bound variables and $x_2, x_5$ as the free variables. Bag $\mB_{t_3}$ contains $x_2, x_5$  as bound variables for $\mB_{t_3}$ and $x_3, x_4$ as free variables.
	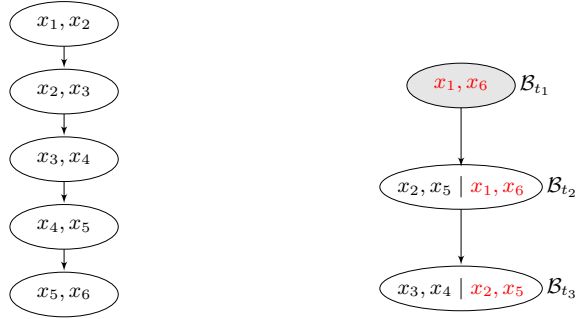
\begin{figure}[t]
		\begin{subfigure}{0.45\linewidth}
					\centering
					\scalebox{.9}{\begin{tikzpicture}
				\tikzset{edge/.style = {->,> = latex'},
					vertex/.style={circle, thick, minimum size=5mm}}
				\def\x{0.25}
				
				\begin{scope}[fill opacity=1]
				
				\draw[] (0,-2) ellipse (0.8cm and 0.33cm) node {\small $x_1, x_2$};
				\draw[] (0,-3) ellipse (0.8cm and 0.33cm) node {\small $x_2, x_3$};
				\draw[] (0,-4) ellipse (0.8cm and 0.33cm) node {\small $x_3, x_4$}; 
				\draw[] (0,-5) ellipse (0.8cm and 0.33cm) node {\small $x_4, x_5$};			  			  
				\draw[] (0,-6) ellipse (0.8cm and 0.33cm) node {\small $x_5, x_6$};
				
				\draw[edge] (0,-2.33) -- (0,-2.65);
				\draw[edge] (0,-3.33) -- (0,-3.65);			  
				\draw[edge] (0,-4.33) -- (0,-4.65);
				\draw[edge] (0,-5.33) -- (0,-5.65);
				\end{scope}	
				\end{tikzpicture}}
			%\caption{Query decomposition for $5$-path query.}
		\label{subfig:1}			
		\end{subfigure}
		\begin{subfigure}{0.45\linewidth}
				\vspace{2em}
					\centering
								\scalebox{.9}{\begin{tikzpicture}
				\tikzset{edge/.style = {->,> = latex'},
					vertex/.style={circle, thick, minimum size=5mm}}
				\def\x{0.25}
				
				\begin{scope}[fill opacity=1]

				\draw[fill=black!10] (5,-2) ellipse (0.8cm and 0.33cm) node {\small ${\color{red} x_1, x_6}$};
				\node[vertex]  at (6.1,-2) {$\bag_{t_1}$};
				\draw[] (5,-3.5) ellipse (1.2cm and 0.33cm) node {\small \small $\prq{x_2, x_5}{\color{red} x_1, x_6}$};
				\node[vertex]  at (6.5,-3.5) {$\bag_{t_2}$};				
				\draw[] (5,-5) ellipse (1.2cm and 0.33cm) node {\small \small $\prq{x_3, x_4}{\color{red} x_2, x_5}$};						
				\node[vertex]  at (6.5,-5) {$\bag_{t_3}$};								
				\draw[edge] (5,-2.33) -- (5,-3.2);
				\draw[edge] (5,-3.83) -- (5,-4.65);			
				\end{scope}	
				\end{tikzpicture} 
				}
			%\caption{$C$-connex decomposition with $C = \{x_1, x_6\}$.} \label{subfig:2}
		\end{subfigure}
		
		\caption{Two tree decompositions for the length-5 path query: the left is unconstrained, while the right is a $C$-connex decomposition with  
			$C = \{x_1, x_6 \}$. The bound variables are colored red. The nodes in $A$ are colored grey.}
		\label{fig:cfhw}	
	\end{figure}
		
\end{example}

Next, we use a parameterized notion of width for the $\nodes_\bound$-connex tree decomposition that was introduced in~\cite{deep2018compressed}. The width is parameterized by a function $\delta$  that maps each node $t$ in the tree to a non-negative number, such that $\delta(t)=0$ whenever $t \in A$. The intuition here is that we will spend $O(|D|^{\delta(t)})$ in the node $t$ while answering the access request. The parameterized width of a bag $\bag_t$ is now defined as:
$ \rho_t(\delta) = \min_{\bu} \left( \sum_{F} u_F - \delta(t) \cdot \slack \right)$
where $\bu$ is a fractional edge cover of the bag $\bag_t$, and $\slack$ is the slack (on the bound variables of the bag). The $\delta$-width of the decomposition is then defined as $\max_{t \notin A} \rho_t(\delta)$. Finally, we define the $\delta$-height as the maximum-weight path from the root to any leaf, where the weight of a path $P$ is $\sum_{t \in P} \delta(t)$. We now have all the necessary machinery to state our second main theorem.

\begin{theorem}\label{thm:main2}
	Let $\varphi^\eta$ be a Boolean adorned query. Consider any $\nodes_\bound$-connex tree decomposition of $\varphi$. For some parametrization $\delta$ of the decomposition, let $f$ be its $\delta$-width, and $h$ be its $\delta$-height. Then, for any input database $D$, we can construct a data structure that answers any access request in time $T = O(|D|^h)$ with space $S = O(|D| + |D|^f)$.
\end{theorem}

The function $\delta$ allows us to trade off between time and space. If we set $\delta(t) = 0$ for every node $t$ in the tree, then the $\delta$-height becomes  $O(1)$, while the $\delta$-width equals to the fractional hypetree width of the decomposition. As we increase the values of $\delta$ in each bag, the $\delta$-height increases while the $\delta$-width decreases, i.e., the answer time $T$ increases while the space decreases. Additionally, we note that the tradeoff from~\autoref{thm:main2} is at least as good as the one from~\autoref{thm:main1}. Indeed, we can always construct a tree decomposition where all variables reside in a single node of the tree. In this case, we recover exactly the tradeoff from~\autoref{thm:main1}.

\begin{toappendix}
\begin{proof}[Proof of \Cref{thm:main2}]
	We introduce some new notation. Let $\tree = (\htree, A)$ denote the $\nodes_\bound$-connex tree decomposition with $f$ as its $\delta$-width, and $h$ as its $\delta$-height. For each node $t \in \htree \setminus A$, we denote by $\manc(t)$  the union of all the bags for the nodes that are the ancestors of $t$, $\nodes_\bound^t = \bag_t \cap \manc(t)$ and $\nodes_\free^t = \bag_t \setminus \nodes_\bound^t$. Intuitively, $\nodes_\bound^t (\nodes_\free^t)$ are the bound (resp. free) variables for the bag $t$ as we traverse the tree top-down.
	For example, in the decomposition on the right hand side of Figure 1, $A = \bag_{t_1}$, $\manc(t_2) = \{x_1, x_6\},$ $\nodes_\bound^t = \{x_1, x_6\}$ and $\nodes_\free^t = \{x_1, x_5\}$. For node $t_3$, $\manc(t_3) = \{x_2, x_5, x_1, x_6\}$, $\nodes_\bound^t = \{x_2, x_5\}$ and $\nodes_\free^t = \{x_3, x_4\}$.
	
	Since we will operate on subtrees of the decomposition, we will use $\bag^\prec_{t}$ to denote the union of all bags in the subtree rooted at $t$ (including $t$). For a given set of variables $I$, we also define $\edges_I = \{ R_i \mid \vars{R_i} \subseteq I\}$ and the join query $J(I) = \big(  \land_{R_i \in \edges_I} R_i(\bx_i)  \big)$ that denotes the join of all atoms that contain only the variables present in $I$. 
	
	 \paragraph*{Data Structure Construction}~\Cref{algo:preprocessing} shows the steps for the data structure construction. We apply \autoref{thm:main1} to each bag (except the root bag) in $\htree$ with the following parameters: $(i)$ fractional edge cover $\mathbf{u}$ corresponding to bag $\bag_t$; $(ii)$ adorned query $\varphi^\eta_t$ corresponding to bag $t$ where the { body of the query} is the join of all atoms that contain { only the variables} present in the bag (i.e. $J(\bag_t)$); and $(iii)$ {head of the  $\varphi^\eta_t$ with bound variables $\nodes_\bound^t = \manc(t) \cap \bag_t$.} \autoref{thm:main1} returns the hash index $\mathcal{K}(t)$ and its space requirement corresponding to bag $\bag_t$ is $S = O(|D| + |D|^{\rho_t(\delta)}) \leq O(|D| + |D|^f)$. This follows directly from the definition of $\delta(t)$ width of bag $t$ which is assumed to be at most $f$. 
	 
	 Next, we create a modified hash index $\mathcal{K'}(t)$ as follows: for each key $\ba$ of the hash index $\mathcal{K}(t)$, we store a yes or no answer to the question whether there exists a valuation for all variables in $\bag^\prec_{t}$ that satisfies the join query $J(\bag^\prec_{t})$ after fixing $\nodes_b^t$ to $\ba$.

  \begin{algorithm}[!t]
	 	\SetCommentSty{textsf}
	 	\SetKwInOut{Input}{Input}
	 	\SetKwInOut{Output}{Output}
	 	\SetKwProg{myproc}{\textsc{procedure}}{}{}
	 	\Input{Input query $\varphi^\eta$, database instance $D$, tree decomposition $\htree$ and edge cover $\bu$ for each node.}
	 	\Output{hash index $\mathcal{K'}(t)$ for each node $t \in \htree$}
	 	\BlankLine
	 	\ForEach{node $t$ in $\htree \setminus A$ }{
	 		$\mathcal{K}(t) \leftarrow $ hash index from \autoref{thm:main1} with cover $\bu$ and $\varphi^\eta_t(\nodes^t_\bound) = \big(  \land_{R_i \in \edges_{\bag_t}} R_i(\bx_i)  \big)$
	 		
	 		$\mathcal{K'}(t) \leftarrow \mathcal{K}(t)$ \tcc*{create a copy}
	 		
	 		\ForEach{$\ba \in \mathcal{K'}(t)$}{
	 			$\gamma^\eta_t(\nodes^t_\bound) = \big(  \land_{R_i \in \edges_{\bag^\prec_{t}}} R_i(\bx_i)  \big)$ \tcc*{query corresponding to subtree rooted at $t$}
	 			
	 			$\mathcal{K'}(t)[a] \leftarrow \gamma^\eta_t(\nodes^t_\bound \leftarrow \ba)$ \tcc*{Boolean answer to whether $\gamma^\eta_t(\nodes^t_\bound \leftarrow \ba)$ has a satisfying assignment when bound variables are fixed to $\ba$}
 			}
	 	}
	 	\BlankLine
	 	\KwRet{$\mathcal{K'}(t)$ for all nodes $t$}
	 	\caption{Data Structure for \autoref{thm:main2}}
	 	\label{algo:preprocessing}
	 \end{algorithm}

  \begin{algorithm}[!t]
	\SetCommentSty{textsf}
	\SetKwInOut{Input}{Input}
	\SetKwInOut{Output}{Output}
	\SetKwProg{myproc}{\textsc{procedure}}{}{}
	\SetKwFunction{enum}{\textsc{findResult}}	
	\SetKwFunction{topdown}{\textsc{Topdown}}
	\SetKwFunction{break}{\textsf{break}}
	\SetKwData{temp}{\textsf{temp}}	
	\Input{Input query $\varphi^\eta$, database instance $D$, tree decomposition $\htree$, hash index $\mathcal{K'}(t)$, access request $\ba$.}
	\Output{$\varphi^\eta(\bx\leftarrow\ba)(D)$}
	\BlankLine
	
	\myproc{\enum{}}{
		\topdown{$A, \ba$} \tcc*{$A$ is the root of $\htree$}
	}

	\myproc{\topdown{$t, v$}}{
		\tcc{$t$ is node in $\htree$ and $v$ is the valuation over bound variables of $t$}
		\If{$t == A$ \textbf{and} $v$ is not valid}{
			\KwRet{\textbf{false}}
		}
		\If{$v \in \mathcal{K'}(t)$}{
			\KwRet{$\mathcal{K'}(t)[v]$} \tcc*{access request present in the hash index}
		}
		$\mL \leftarrow $result of WCOJ with $\bu$ for $t$ on $\varphi^\eta_t$
		
		\ForEach{$w \in \mL$}{
			\temp $\leftarrow $\textbf{ true}
			
			\ForEach{child $t'$ of $t$}{
				\If{\topdown{$t', \pi_{\nodes^{t'}_\bound}(w)$} is \textbf{false}}{
					\temp $\leftarrow $\textbf{false} \tcc*{If a $w$ does not lead to a join for any child $t'$, break and continue search}
					\break
				}
			}
			\If{\temp  is \textbf{true}}{
				\KwRet{\textbf{true}} \tcc*{found a $w$ that leads to a join for all children of $t$}
			}
		}
		\KwRet{\textbf{false}} \tcc*{couldn't find any $w$}
	}
	\caption{Query answering for \autoref{thm:main2}}
	\label{algo:enumerate}
\end{algorithm}

	\paragraph*{Query Answering} We now describe the query answering algorithm. Let $C = \{x_1, \dots , x_k\}$ and access request $\mathbf{a} = (a_1, \dots, a_k)$. We first need to check whether $\mathbf{a}$ is valid. If the request is not valid, we can simply output no. This can be done in constant time after creating hash indexes of size $O(|D|)$ on the base relations during the preprocessing phase. If the access request is valid, the second step is to check whether $Q(\mathbf{a})$ is true or false. Let $\mathcal{P}$ denote the set of bags that are children of root bag. Then, for each bag $\bag_t \in \mathcal{P}$, we check whether $\pi_{\nodes_\bound^t} (\mathbf{a}) \in \mathcal{K'}(t)$. If it is stored, it means that that running time of $\pi_{\nodes_\bound^t} (\mathbf{a})$ is greater than $O(|D|^{\delta(t)})$. If the entry for $\pi_{\nodes_\bound^t} (\mathbf{a}) $ is false in the data structure, we can return false\footnote{we return to the caller of the function since the query answering algorithm is recursive.} since we know that no output tuple can be formed by the subtree rooted at $\bag_t$. Similarly, if the entry for $\pi_{\nodes_\bound^t} (\mathbf{a}) $ is true, we can immediately return true since we know that there is at least one valid tuple that can be formed by the subtree rooted at $\bag_t$.
	
	If there is no entry for $\pi_{\nodes_\bound^t} (\mathbf{a})$ in $\mathcal{K'}(t)$, this means that answering time of evaluating the join at node $t$ is $T \leq O(|D|^{\delta(t)})$. Thus, we can evaluate the join for the bag by fixing $\nodes_\bound^t$ as $\pi_{\nodes_\bound^t} (\mathbf{a})$ using any worst-case optimal join algorithm (WCOJ), which guarantees that the total running time is at most $O(|D|^{\delta(t)})$. If no output is generated, the algorithm outputs false since no output tuple can be formed by subtree rooted at $\bag_t$. If there is output generated, then there can be at most $O(|D|^{\delta(t)})$ tuples at bag $\bag_t$ over the set of all non-bound variables in the bag. For each of these tuples, we recursively proceed to the children of bag $\bag_t$ and repeat the algorithm as shown in~\Cref{algo:enumerate}. Each fixing of variables at $\bag_t$ acts as the bound variables for the children bag. In the worst case, all bags in $\htree$ may require join processing. Since the query size is a constant, it implies that the number of root to leaf paths are also constant. Thus, the answering time is dominated by the  root to leaf path with the largest $\delta_t$ sum, i.e the $\delta$-height of the decomposition. Thus, $T = O(|D|^{\sum_{t \in P} \delta(t)}) = O(|D|^h)$.
\end{proof}
\end{toappendix}

\begin{example} \label{example:5path}
	We continue with the $5$-path query. Since $\bag_{t_1} = \{x_1, x_6\} \in A$, we assign $\delta(t_1)=0$. For $\bag_{t_2} = \{x_1,x_2,x_5,x_6\}$, the only valid fractional edge cover assigns weight 1 to both $R_1,R_5$ and has slack 1. Hence, if we assign $\delta(t_2) = \tau$ for some parameter $\tau$, the width is $2-\tau$. For $\bag_{t_3} = \{x_2,x_3,x_4, x_5\}$, the only fractional cover also assigns weight 1 to both $R_2, R_4$, with slack $1$ again. Assigning $\delta(t_3) = \tau$, the width becomes $2-\tau$ for $t_3$ as well. 	Hence, the $\delta$-width of the tree decomposition is $2-\tau$, while the $\delta$-height is $2 \tau$. Plugging this to~\autoref{thm:main2}, it gives us a tradeoff with answering time $T = O(|E|^{2\tau})$ and space usage $S = O(|E|+|E|^{2-\tau})$, which matches the state-of-the-art result in~\cite{goldstein2017conditional}. 
\end{example}

For the $k$-reachability problem, a general tradeoff $S \times T^{2/(k-1)} = O(|D|^2)$ was also shown by~\cite{goldstein2017conditional} using a careful recursive argument. The data structure generated using~\autoref{thm:main2} is able to recover the tradeoff. In particular, we obtain the answering time as $T = O(|E|^{(k-1)\tau/2})$ using space $S = O(|E|+|E|^{2-\tau})$.

\begin{example} Consider a variant of the square detection problem: given two vertices, the goal is to decide whether they occur in two opposites corners of a square, which can be captured by the following adorned query: 
$$ \varphi^{\bound \bound}(x_1,x_3) =  R_1(x_1, x_2) \land R_1(x_2, x_3) \land R_3(x_3, x_4) \land R_4(x_4,x_1).$$
\autoref{thm:main1} gives a tradeoff with answering time $O(T)$ and space $O(|E|^2/T)$. But we can obtain a better tradeoff using~\autoref{thm:main2}. Indeed, consider the tree decomposition where we have a root bag $t_1$ with $\bag_{t_1} = \{x_1,x_3\}$, and two children of $t_1$ with Boolean $\bag_{t_2} = \{x_1,x_2, x_3\}$ and $\bag_{t_3} = \{x_1,x_3, x_4\}$. For $\bag_{t_2}$, we can see that if assigning a weight of $1$ to both hyperedges, we get a slack of $2$. Hence,  if $\delta(t_2) = \tau$, the $\delta$-width is $2-2\tau$. Similarly for $t_3$, we assign $\delta(t_3) = \tau$, for a $\delta$-width with $2-2\tau$. Applying~\autoref{thm:main2}, we obtain a tradeoff with time $T = O(|E|^\tau)$ (since both root-leaf paths have only one node), and space $S = O(|E| + |E|^{2-2\tau})$. So the space usage can be improved from $O(|E|^2/T)$ to $O(|E|^2/T^2)$.
\end{example}
 
\section{CQs with Negation} \label{sec:negation}

In this section, we present a simple but powerful extension of our result to adorned Boolean CQs with negation. 
A CQ with negation, denoted as $CQ^\neg$, is a CQ where some of the atoms can be negative, i.e., $\neg R_i(\bx_i)$ is allowed. For $\varphi \in CQ^\neg$, we denote by $\varphi^+$ the conjunction of the positive atoms in $\varphi$ and $\varphi^-$ the conjunction of all negated atoms. A $CQ^\neg$ is said to be \emph{safe} if every variable appears in at least some positive atom. In this paper, we restrict our scope to the class of safe $CQ^\neg$, a standard assumption~\cite{wei2003containment,nash2004processing} ensuring that query results are well-defined and do not depend on domains.

{Given a query $\varphi \in CQ^\neg$, we build the data structure from~\autoref{thm:main2} for $\varphi^+$ but impose two constraints on the decomposition: $(i)$ no leaf node(s) contains any free variables, $(ii)$ for every negated atom $R^-$, all variables of $R^-$ must appear together as bound variables in some leaf node(s). In other words, there exists a leaf node such that $\vars{R^-}$ is present in it. It is easy to see that such a decomposition always exists. Indeed, we can fix the root bag to be $C = \mathcal{V}_\bound$, its child bag with free variables as $\vars{\varphi^+} \setminus C$ and bound variables as $C$, and the leaf bag, which is connected to the child of the root, with bound variables as $\vars{\varphi^-}$ without free variables. Observe that the bag containing $\vars{\varphi^+}$ free variables can be covered by only using the positive atoms since $\varphi$ is safe. The intuition is the following: during the query answering phase, we wish to find the join result over all variables $\mathcal{V}_\free$ before reaching the leaf nodes; and then, we can check whether there the tuples satisfy the negated atoms or not, in $O(1)$ time. The next example shows the application of the algorithm to adorned path queries containing negation.
	
	\begin{example} \label{ex:2}
		Consider the query $Q^{\bound \bound}(x_1, x_6) =  R(x_1, x_2) \land \neg S(x_2, x_3) \land T(x_3, x_4) \land \neg U(x_4,x_5) \land$ $V(x_5, x_6)$. Using the decomposition in~\autoref{fig:decom}, we can now apply~\autoref{thm:main2} to obtain the tradeoff $S = O(|D|^3/\tau)$ and $T = O(\tau)$. Both leaf nodes only require linear space since a single atom covers the variables. Given an access request, we check whether the answer for this request has been materialized or not. If not, we proceed to the query answering phase and find at most  $O(\tau)$ answers after evaluating the join in the middle bag. For each of these answers, we can now check in constant time whether the tuples formed by values for $x_2, x_3$ and $x_4, x_5$ are not present in relations $S$ and $U$ respectively.
	\end{example}
	
	For adorned queries where $\mathcal{V}_\bound \subseteq \vars{\varphi^-}$, we can further simplify the algorithm. In this case, we no longer need to create a constrained decomposition since the check to see if the negated relations are satisfied or not can be done in constant time at the root bag itself. Thus, we can directly build the data structure from~\autoref{thm:main2} using the query $\varphi^+$.

\begin{example}[Open Triangle Detection] \label{ex:3} Consider the query $\varphi^{\bound \bound}(x_2, x_3) =$ $R_1(x_1, x_2)$ $ \land \neg R_2(x_2, x_3) \land R_3(x_1, x_3)$, where $\varphi^-$ is $\neg R_2(x_2, x_3)$ and $\varphi^+$ is $R_1(x_1, x_2) \land R_3(x_1, x_3)$ with the adorned view as $\varphi^{+\bound \bound}(x_2, x_3) =  R_1(x_1, x_2) \land R_3(x_1, x_3)$. Observe that $\{x_2, x_3\} \subseteq \vars{\varphi^-}$. We apply~\autoref{thm:main2} to obtain the tradeoff $S = O(|E|^2/ \tau^2)$ and $T = O(\tau)$ with root bag $C = \{x_2, x_3\}$, its child bag with $\mathcal{V}_\bound = C$ and $\mathcal{V}_\free = \{x_1\}$, and the leaf bag to be $\mathcal{V}_\bound = C$ and $\mathcal{V}_\free = \emptyset$. Given an access request $(a,b)$, we check whether the answer for this request has been materialized or not. If not,  we traverse the decomposition and evaluating the join to find if there exists a connecting value for $x_1$. For the last bag, we simply check whether $(a,b)$ exists in $R_2$ or not  in $O(1)$ time.
\end{example}

\introparagraph{A note on optimality} It is easy to see that the algorithm obtained for Boolean CQs with negation is conditionally optimal assuming the optimality of~\autoref{thm:main2}. Indeed, if all negated relations are empty, the join query is equivalent to $\varphi^+$ and the algorithm now simply applies~\autoref{thm:main2} to $\varphi^+$. In ~\autoref{ex:3}, assuming relation $R_2$ is empty, the query is equivalent to set intersection whose tradeoffs are conjectured to be optimal.

}

\begin{figure}    
\centering
	\scalebox{0.85}{
	
	\tikzset{every picture/.style={line width=0.75pt}} %set default line width to 0.75pt        
	\centering
	\begin{tikzpicture}[x=0.75pt,y=0.75pt,yscale=-1,xscale=1]
		%uncomment if require: \path (0,283); %set diagram left start at 0, and has height of 283
		
		%Shape: Ellipse [id:dp40807706596879334] 
		\draw[fill=black!10]   (267,49) .. controls (267,37.95) and (288.71,29) .. (315.5,29) .. controls (342.29,29) and (364,37.95) .. (364,49) .. controls (364,60.05) and (342.29,69) .. (315.5,69) .. controls (288.71,69) and (267,60.05) .. (267,49) -- cycle ;
		%Shape: Ellipse [id:dp40303668432342676] 
		\draw   (215,114) .. controls (215,102.95) and (260.22,94) .. (316,94) .. controls (371.78,94) and (417,102.95) .. (417,114) .. controls (417,125.05) and (371.78,134) .. (316,134) .. controls (260.22,134) and (215,125.05) .. (215,114) -- cycle ;
		%Shape: Ellipse [id:dp11366463465130394] 
		\draw   (337,180) .. controls (337,168.95) and (358.71,160) .. (385.5,160) .. controls (412.29,160) and (434,168.95) .. (434,180) .. controls (434,191.05) and (412.29,200) .. (385.5,200) .. controls (358.71,200) and (337,191.05) .. (337,180) -- cycle ;
		%Shape: Ellipse [id:dp3108792880326101] 
		\draw   (205,179) .. controls (205,167.95) and (226.71,159) .. (253.5,159) .. controls (280.29,159) and (302,167.95) .. (302,179) .. controls (302,190.05) and (280.29,199) .. (253.5,199) .. controls (226.71,199) and (205,190.05) .. (205,179) -- cycle ;
		%Straight Lines [id:da06001038742055842] 
		\draw[->]    (315.5,69) -- (316,94) ;
		%Straight Lines [id:da6343717403721572] 
		\draw[->]   (314,134) -- (253.5,159) ;
		%Straight Lines [id:da07453323475638762] 
		\draw[->]    (314,134) -- (388.5,160) ;
		
		% Text Node
		\draw (300,45) node [anchor=north west][inner sep=0.75pt]   [align=left] {$\displaystyle \color{red}{x_{1} ,x_{6}}$};
		% Text Node
		\draw (240,104) node [anchor=north west][inner sep=0.75pt]   [align=left] {$\displaystyle \textcolor[rgb]{0,0,0}{\ x_{2} ,\ x_{3} ,\ x_{4} ,\ x_{5} \ |}\color{red} { \hspace{0.6em} x_1, x_6}$};
		% Text Node
		\draw (233,169) node [anchor=north west][inner sep=0.75pt]   [align=left] {$\displaystyle \textcolor[rgb]{0,0,0}{|} \color{red} {\hspace{0.6em} x_2, x_3}$};
		% Text Node
		\draw (371,171) node [anchor=north west][inner sep=0.75pt]   [align=left] {$\displaystyle \textcolor[rgb]{0,0,0}{|} \color{red} { \hspace{0.6em} x_4, x_5}$};
	\end{tikzpicture}}
	\caption{$C$-connex decomposition for~\autoref{ex:2}.} \label{fig:decom}
\end{figure}
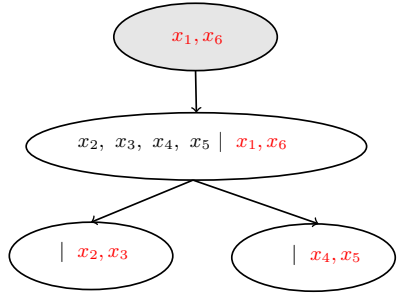
% !TeX root = paper.tex
\section{Path Queries} \label{sec:path}

In this section, we present an algorithm for the adorned query $P_k^{\bound \bound}(x_1, x_{k+1}) =  R_1(x_1,x_2) \land \dots \land R_k(x_k, x_{k+1})$ that improves upon the conjectured optimal solution. Before  diving into the details, we first state the upper bound on the tradeoff between space and query time.

\begin{theorem} [due to \cite{goldstein2017conditional}] \label{thm:path:existing}
	There exists a data structure for solving $P_k^{\bound \bound}(x_1, x_{k+1})$ with space $S$ and answering time $T$ such that $S \cdot T^{2/(k-1)} = O(|D|^2)$.
\end{theorem}

Note that for $k=2$, the problem is equivalent to \textsf{SetDisjointness} with the space/time tradeoff as $S \cdot T^2 = O(N^2)$.%~\cite{goldstein2017conditional} proved that this bound is conditionally optimal based on the strong $3$-\textsf{SUM Indexing} conjecture which was subsequently shown to be false. 
~\cite{goldstein2017conditional} also conjectured that the tradeoff is essentially optimal.

\begin{conjecture}[due to \cite{goldstein2017conditional}] \label{conjecture:path}
	Any data structure for $P_k^{\bound \bound}(x_1, x_{k+1})$ with answering time $T$ must use space $S = \tilde{\Omega}(|D|^2/T^{2/(k-1)})$.
\end{conjecture}

Building upon~\autoref{conjecture:path},~\cite{goldstein2017conditional} also showed a result on the optimality of approximate distance oracles.  Our result implies that~\autoref{thm:path:existing} can be improved further, thus refuting ~\autoref{conjecture:path}.  The first observation is that the tradeoff in~\autoref{thm:path:existing} is only useful when $T \leq |D|$. Indeed, we can always answer any Boolean path query in linear time using breadth-first search. 
Surprisingly, it is also possible to improve~\autoref{thm:path:existing} for the regime of small answering time as well. In what follows, we will show the improvement for paths of length 4; we will generalize the algorithm for any length later.

\subsection{Length-4 Path}

\begin{lemma} \label{lem:basic}
	There exists a parameterized data structure for solving $P_4^{\bound \bound}(x_1, x_{5})$ that uses space $S$ and answering time $T \leq \sqrt{|D|}$ that satisfies the tradeoff $S \cdot T = O(|D|^2)$.
\end{lemma}

For $k=4$,~\autoref{thm:path:existing} gives us the tradeoff $S \cdot T^{2/3} = O(|D|^2)$ which is always worse than the tradeoff in~\autoref{lem:basic}. We next present our algorithm in detail.

\smallskip
\introparagraph{Preprocessing Phase} Consider $P_4^{\bound \bound}(x_1, x_{5}) =  R(x_1, x_2) \land S(x_2, x_3) \land T(x_3, x_4) \land U(x_4, x_5)$. Let $\Delta$ be a degree threshold. We say that a constant $a$ is {\em heavy} if its frequency on attribute $x_3$ is greater than $\Delta$ in both relations $S$ and $T$; otherwise, it is {\em light}. In other words, $a$ is heavy if $|\sigma_{x_3=a}(S)| > \Delta$ and $|\sigma_{x_3=a}(T)| > \Delta$. We distinguish two cases based on whether a constant for $x_3$ is heavy or light. Let $\mL_{\textsf{heavy}}(x_3)$ denote the unary relation that contains all heavy values, and $\mL_{\textsf{light}}(x_3)$ the one that contains all light values. Observe that we can compute both of these relations in time $O(|D|)$ by simply iterating over the active domain of variable $x_3$ and checking the degree in relations $S$ and $T$.  We compute  two views:
\begin{align*}
    V_1(x_1,x_3) &= R(x_1, x_2) \land S(x_2, x_3) \land \mL_{\textsf{heavy}}(x_3) \\
    V_2(x_3,x_5) &=  \mL_{\textsf{heavy}}(x_3) \land T(x_3, x_4) \land U(x_4, x_5)
\end{align*}

We store the views as a hash index that, given a value of $x_1$ (or $x_5$), returns all matching values of $x_3$.
Both views take space $O(|D|^2/ \Delta)$. Indeed, $|\mL_{\textsf{heavy}}| \leq  |D|/\Delta$. Since we can construct a fractional edge cover for $V_1$ by assigning a weight of 1 to $R$ and $\mL_{\textsf{heavy}}$, this gives us an upper bound of $|D| \cdot (|D|/\Delta)$ for the query output. The same argument holds for $V_2$. We also compute the following view for light values:  $V_3(x_2, x_4) =   S(x_2, x_3) \land \mL_{\textsf{light}}(x_3) \land T(x_3, x_4).$ This view requires space $O(|D| \cdot \Delta)$, since the degree of the light constants is at most $\Delta$ (i.e. $\sum_{x \in \mL_{\textsf{light}}(x_3)} |S(x_2,x) \land T(x, x_4)| \leq \sum_{x \in \mL_{\textsf{light}}(x_3)} |S(x_2,x)| \cdot |T(x, x_4)| \leq \sum_{x \in \mL_{\textsf{light}}(x_3)} |S(x_2,x)| \cdot \Delta \leq |D| \cdot \Delta$). We can now rewrite the original query  as  $P_4^{\bound \bound}(x_1, x_{5}) =  R(x_1, x_2) \land V_3(x_2,x_4) \land U(x_4, x_5).$

The rewritten query is a three path query. Hence, we can apply~\autoref{thm:main1} to create a data structure with answering time $T = O(|D|/\Delta)$ and space $S = O(|D|^2/(|D|/\Delta) ) = O(|D| \cdot \Delta)$.

\smallskip
\introparagraph{Query Answering} Given an access request, we first check whether there exists a 4-path that goes through some heavy value in $\mL_{\textsf{heavy}}(x_3)$. This can be done in time $O(|D|/\Delta)$ using the views $V_1$ and $V_2$. Indeed, we obtain at most $O(|D|/\Delta)$ values for $x_3$ using the index for $V_1$, and $O(|D|/\Delta)$ values for $x_3$ using the index for $V_3$. We then intersect the results in time $O(|D|/\Delta)$ by iterating over the $O(|D|/\Delta)$ values for $x_3$ and checking if the bound values for $x_1$ and $x_5$ from a tuple in $V_1$ and $V_2$ respectively.
If we find no such 4-path, we check for a 4-path that uses a light value for $x_3$. From the data structure we have constructed in the preprocessing phase, we can do this in time $O(|D|/ \Delta)$.

\smallskip
\introparagraph{Tradeoff Analysis}
From the above, we can compute the answer in time $T = O(|D|/ \Delta)$. From the analysis in the preprocessing phase, the space needed is $S = O(|D|^2/\Delta + |D| \cdot \Delta)$. Thus, whenever $\Delta \geq \sqrt{|D|}$, the space becomes $S = O(|D| \cdot \Delta)$, completing our analysis.

\subsection{General Path Queries}

We can now use the algorithm for the 4-path query to improve the space-time tradeoff for general path queries of length greater than four. 

\begin{theoremrep} \label{thm:path}
Let $D$ be an input instance.
	For $k \ge4$, there is a data structure for $P_k^{\bound \bound}(x_1, x_{k+1})$ with space $S = O(|D|\cdot \Delta)$ and answer time $T= O\left((\frac{|D|}{\Delta})^{\frac{k-2}{2}}\right)$ for $\Delta \geq \sqrt{|D|}$.
\end{theoremrep}
\begin{proof}
	Fix some $\Delta \geq \sqrt{|D|}$.  We construct the data structure for a path of length $k$ recursively. The base case is when $k=4$, with answer time $T_4 = |D|/\Delta$ and space $S_4 = |D| \cdot \Delta$. 
	
	In the recursive step, similar to the previous section, we set $\sqrt{\frac{|D|}{\Delta}}$ as the degree threshold for any constant that variables $x_{1}$ and $x_{k+1}$ can take. Let $\mL_{\textsf{heavy}}^1, \mL_{\textsf{heavy}}^{k+1}$ be unary relations that store the heavy values for $x_1, x_{k+1}$ respectively. We compute and store the result of 
	$$V(x_1, x_{k+1}) = \mL_{\textsf{heavy}}^1(x_1) \land R_1(x_1, x_2) \land \dots \land R_k(x_k, x_{k+1}), \mL_{\textsf{heavy}}^{k+1}(x_{k+1}).$$
	This view has size bounded by $\left(|D|\cdot \sqrt{\frac{\Delta}{|D|}}\right)^2 = |D| \cdot \Delta$. We consider the following queries:
	$$ V_1^{\bound\bound}(x_2, x_{k+1}) =  R_2(x_2, x_3) \land \dots \land R_k(x_k, x_{k+1}). $$
	$$V_2^{\bound\bound}(x_1, x_{k}) =  R_1(x_1, x_2) \land \dots \land R_{k-1}(x_{k-1}, x_{k}).$$
	both of which correspond to the $(k-1)$-path, so we can recursively apply the data structure here. Let $S_k, T_k$ be the space and time for $k$-path. For space, we have following observation:	%
	$$ S_k = |D| \cdot \Delta + S_{k-1}$$ 
	As $S_4 = |D|\cdot \Delta$, we obtain $S_k = O(|D|\cdot \Delta)$. 
	
	Given an access request, we answer it by distinguishing two cases. If $x_1,x_{k+1}$ is heavy, we probe the stored view $V(x_1,x_{k+1})$ in time $O(1)$. If one of them is light (say w.l.o.g. $x_1$), we call recursively the data structure $V_1$ for every one of the $\leq \sqrt{|D|/\Delta}$ values connected with $x_1$. This gives us the following recurrence formula for answer time:
	$$ T_k = (|D|/\Delta)^{1/2} \cdot T_{k-1}$$ 
	Solving the recursive formula gives us $T_k = (|D|/\Delta)^{(k-2)/2}$.
\end{proof}

The space-time tradeoff obtained from ~\autoref{thm:path} is $S \cdot T^{2/(k-2)} = O(|D|^2)$, but only for $T \leq |D|^{(k-2)/4}$. To compare it with the tradeoff of $S \cdot T^{2/(k-1)} = O(|D|^2)$ obtained from~\autoref{thm:path:existing}, it is instructive to look at Figures~\ref{fig:k4} and~\ref{fig:k6}, which plot the space-time tradeoffs for $k=4$ and $k=6$ respectively. In general, as $k$ grows, the new tradeoff line (labeled as $\rho_1$) becomes flatter and approaches \autoref{thm:path:existing}.

\def\ra{1.5}
\begin{figure*}[!htp]
    %\centering
    \begin{subfigure}[b]{0.48\linewidth} 
    \scalebox{0.7}{
    \begin{tikzpicture}[scale=2]
    % Draw axes
    \draw [<->,thick] (0,2.2) node (yaxis) [above] {}
        |- (1.72*\ra,0) node (xaxis) [below] {};
     \node at (0.3,2.4) {$\log_{|D|}(S)$};   
     \node at (1.85*\ra,-0.2) { $\log_{|D|}(T)$};
    % prior tradeoff    
    \draw[brown] (0,2) coordinate (A) -- (1.5*\ra,0) coordinate (B);
    \draw[red] (0,2) coordinate (b_1) -- (0.5*\ra,1) coordinate (b_2) node[below, xshift=-1em] {$\rho_1$};
    \draw[dotted,very thick,red] (b_2) -- (3*\ra/4,1);
    \draw[densely dotted,very thick,blue] (1*\ra,2/3) -- (1*\ra,0) node[right, yshift=1.5em] {$\rho_2$};
    \draw[dashed] (0,2/3) node[left] {$4/3$} -- (1*\ra,2/3);
    \node at (1*\ra,1.33) (r4) {{\color{brown} \textit{baseline}}};
    \coordinate (c) at (1.5*\ra,2);
    \node at (-0.1,0) {$1$};
    \node at (0,-0.18) {$0$};
    \node at (1.5,-0.18) {$1$};
    \draw[dashed] (yaxis |- b_2) node[left] {$3/2$}
        -| (xaxis -| b_2) node[below] {$1/2$};
    \draw[dashed] (yaxis |- c) node[left] {$2$}
        -| (xaxis -| c) node[below] {$3/2$};
      \coordinate (u) at (1*\ra,1);
    %\draw[dashed] (yaxis |- u) node[left] {$1$}
    %    -| (xaxis -| u) node[below] {$1$};       
    %\draw[dashed] (1,2) -- (1,0) node[below] {$1$};    
    \end{tikzpicture}}
    \caption{4-reachability CQAP.}
 \label{fig:k4}
\end{subfigure}
\hspace{-3em}
\begin{subfigure}[b]{0.48\linewidth}
\scalebox{0.7}{
    \begin{tikzpicture}[scale=2]
    % Draw axes
    \node at (-0.1,0) {$1$};
    \node at (0,-0.18) {$0$};
    \draw [<->,thick] (0,2.2) node (yaxis) [above] {}
        |- (2.7*\ra,0) node (xaxis) [below] {};
     \node at (0.3,2.4) {$\log_{|D|}(S)$};   
     \node at (2.85*\ra,-0.2) { $\log_{|D|}(T)$};
    % prior tradeoff    
    \draw[brown] (0,2) coordinate (A) -- (2.5*\ra,0) coordinate (B);
    \draw[red] (0,2) coordinate (b_1) -- (1*\ra,1) coordinate (b_2) node[below, xshift=-1em] {$\rho_1$};
    %\draw[dotted,very thick,red] (b_2) -- (1/2*\ra,3/2);
    \draw[densely dotted,very thick,blue] (1*\ra,1) -- (1*\ra,0) node[right, yshift=3em] {$\rho_2$};
    \node at (1*\ra,-0.18) {$1$};
    \draw[dashed] (0,1) node[left] {$3/2$} -- (1*\ra,1);
    \node at (1.5*\ra,1.33) (r4) {{\color{brown} \textit{baseline}}};
    \coordinate (c) at (2.5*\ra,2);
    \draw[dashed] (yaxis |- c) node[left] {$2$}
        -| (xaxis -| c) node[below] {$5/2$};
      \coordinate (u) at (1*\ra,1);
    \end{tikzpicture}}
\caption{6-reachability CQAP.}
 \label{fig:k6}
\end{subfigure}  
\caption{Space/time tradeoffs for path query of length $k \in \{4,6\}$. The line in brown (baseline) shows the tradeoff obtained from~\autoref{thm:path:existing}. The red curve ($\rho_1$) is the new tradeoff obtained using~\autoref{thm:path} and $\rho_2$ shows the transition to when BFS takes over as the best algorithm.}
\end{figure*}
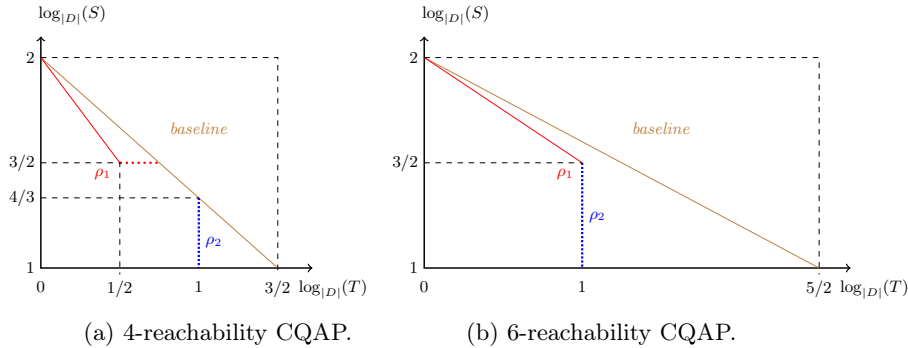

% !TeX root = paper.tex
\section{Lower Bounds} \label{sec:lowerbound}

In this section, we study the lower bounds for adorned star and path queries. We first present conditional lower bounds for the $k$-\setdisj problem using the conditional optimality of $\ell$-\setdisj where $\ell < k$. First, we review the known results from~\cite{goldstein2017conditional} starting with the conjecture for $k$-\setdisj.

\begin{conjecture}[due to \cite{goldstein2017conditional}] \label{conjecture:star}
	Any data structure for $k$-\setdisj problem that answers queries in time $T$ must use space $S = \Omega(|D|^k/T^k)$.
\end{conjecture}

\autoref{conjecture:star} was shown to be conditionally optimal based on conjectured lower bound for the $(k+1)$-\textsf{Sum Indexing} problem, however, it was subsequently showed to be false~\cite{kopelowitz2019strong}, which implies that~\autoref{conjecture:star} is still an open problem. \autoref{conjecture:star} can be further generalized to the case when input relations are of unequal sizes as follows.

\begin{conjecture} \label{conjecture:star:unequal:relation}
	Any data structure for $ \varphi_*^{\bound \dots \bound}(y_1, \dots, y_k) =  R_1(x,y_1) \land \dots \land R_k(x,y_k)$ that answers queries in time $T$ must use space $S = \Omega(\Pi_{i=1}^k |R_i|/T^k)$.
\end{conjecture}

We now state the main result for star queries.

\begin{theoremrep} \label{thm:lower:bound}
	Suppose that any data structure for $ \varphi_*^{\bound \dots \bound}(y_1, \dots, y_k)$ with answering time $T$ must use space $S = \Omega(\Pi_{i=1}^{k} |R_i|/T^k)$. Then, any data structure for $Q_*^{\bound \dots \bound}(y_1, \dots, y_\ell)$ with answering time $T$ must use space  $S = \Omega(\Pi_{i=1}^{\ell} |R_i|/T^\ell)$, for $2 \leq \ell < k$.
\end{theoremrep}
\begin{proof}
	Let $\Delta = T$ be the degree threshold for the $k$ bound variables $y_1, \dots y_k$. If any of the $k$ variables is light (i.e $|\sigma_{y_i=a[y_i]} R_i(y_i, x)| \leq \Delta$), then we can check whether the intersection between $k$ sets is empty or not in time $O(T)$ by indexing all relations in a linear time preprocessing phase. The remaining case is when all $k$ variables are heavy. We now create $\ell$ views $V_1, \dots V_\ell$ by arbitrarily partitioning the $k$ relations into the $\ell$ views followed by materializing the join of all relations in each view. Let view $V_i$ contain the join of $k_i$ relations. Then, $|V_i| = O((\Pi_{R \in J_i} |R|/T^{{k_i}-1}))$ where $J_i$ is the set of all relations assigned to view $V_i$.
	
	We have now reduced the $k$-star query where all $k$ variables are heavy into an instance of $\ell$-star query where the input relations are $V_1, \dots, V_\ell$. Suppose that there exists a data structure that can answer queries in time $T$ using space $S = o(\Pi_{i=1}^{\ell} |V_i|/T^\ell)$. Then, we can use such a data structure for answering the original query where all variables are heavy. The space used by this oracle is 
	\begin{align*}
		\hspace{8em} S &= o(\Pi_{i=1}^{\ell} |V_i|/T^\ell) = o((\Pi_{i=1}^\ell \Pi_{R \in J_i} |R|/T^{{k_i}-1}) \cdot (1 / T^\ell)) \\
		&= o((\Pi_{i=1}^{k} |R_i| / T^{k - \ell}) \cdot (1 / T^\ell)) = o(\Pi_{i=1}^{k} |R_i|/T^k)
	\end{align*}
	which contradicts the space lower bound for $k$-star.
\end{proof}

\autoref{thm:lower:bound} creates a hierarchy for $k$-\setdisj, where the optimality of smaller set disjointness instances depends on larger set disjointness instances. Next, we show conditional lower bounds on the space requirement of path queries. We begin by proving a simple result for optimality of $P^{\bound \bound}_2$ (equivalent to $2$-Set Disjointness) assuming the optimality of $P^{\bound \bound}_3$ query.

\begin{theoremrep} \label{thm:lb:path}
	Suppose that any data structure for $P^{\bound \bound}_3$ that answers queries in time $T$, uses space $S$ such that $S \cdot T = \Omega(|D|^2)$. Then, for $P^{\bound \bound}_2$ , for any data structure that uses space $S = O(|D|^2/T^2)$, the answering time is $\Omega(T)$.
\end{theoremrep}
\begin{proof}
	Let $\Delta = T$ be the degree threshold for all vertices. If both bound variables in $P^{\bound \bound}_3$ are heavy, then we can answer the query in constant time using space $\Theta(|D|^2/T^2)$ by materializing the answers to all heavy-heavy queries. In the remaining cases, at least one of the bound valuations is light. Without loss of generality, suppose $x_1$ is light. Then, we can make $\Delta$ calls to the oracle for query $P^{\bound \bound}_2(x_2, x_4) = R_2(x_2, x_3), R_3(x_3, x_4)$.
	
	Suppose that there exists a data structure with space $O(|D|^2/T^2)$ for $P^{\bound \bound}_2(x_2, x_4)$ and answering time $o(T)$. Then, we can answer $P^{\bound \bound}_3$ with light $x_1$ in time $o(T^2)$. This improves the tradeoff for $P^{\bound \bound}_3$ since the product of space usage and answering time is $o(|D|^2)$ for any non-constant $T$, coming to a contradiction.
\end{proof}

Using a similar argument, it can be shown that the conditional optimality of~\autoref{thm:path} for $k=4$ implies that $S \cdot T = \Omega(|D|^2)$ tradeoff for $P^{\bound \bound}_3$ is also optimal (but only for the range $T \leq \sqrt{|D|}$ when the result is applicable).
\section{Related Work} \label{sec:related}

The study of fine-grained space/time tradeoffs for query answering is a relatively recent effort in the algorithmic community. The study of distance oracles over graphs was first initiated by~\cite{patrascu2010distance} where lower bounds are shown on the size
of a distance oracle for sparse graphs based on a conjecture about the best possible data structure for a set intersection problem.~\cite{cohen2010hardness} also considered the problem of set intersection and presented a data structure that can answer boolean set intersection queries which is conditionally optimal~\cite{goldstein2017conditional}. There also exist another line of work that looks at the problem of approximate distance oracles. Agarwal et al.~\cite{agarwal2011approximate,agarwal2014space} showed that for stretch-2 and stretch-3 oracles, we can achieve $S \times T = O(|D|^2)$ and $S \times T^2 = O(|D|^2)$. They also showed that for any integer $k$, a stretch-$(1+1/k)$ oracle exhibits $S \times T^{1/k} = O(|D|^2)$ tradeoff. Unfortunately, no lower bounds are known for non-constant query time. The authors in~\cite{goldstein2017conditional} conjectured that the tradeoff $S \times T^{2/(k-1)} = O(|D|^2)$ for $k$-reachability is optimal which would also imply that stretch-$(1+1/k)$ oracle tradeoff is also optimal. A different line of work has considered the problem of enumerating query results~\cite{segoufin2013enumerating} of a non-boolean query.~\cite{cohen2010hardness} presented a data structure to enumerate the intersection of two sets with guarantees on the total answering time. This result was generalized to incorporate {\em full} adorned views over CQs~\cite{deep2018compressed}. Our work extends the results to the setting where the join variables are projected away from the query result (i.e. the adorned views are {\em non-full}) and makes the connection between several different algorithmic problems that have been studied independently. Further, we also consider boolean CQs that may contain negations. In the non-static setting,~\cite{berkholz2017answering} initiated the study of answering conjunctive query results under updates. More recently,~\cite{kara2019counting} presented an algorithm for counting the number of triangles under updates. There have also been some exciting developments in the space of enumerating query results with delay for a proper subset of CQs known as {\em hierarchical queries}.~\cite{kara19} presented a tradeoff between preprocessing time and delay for enumerating the results of any (not necessarily full) hierarchical queries under static and dynamic settings. It remains an interesting problem to find improved algorithms for more restricted set of CQs such as hierarchical queries.

\section{Conclusion} \label{sec:conclusion}

In this paper, we investigated the tradeoffs between answering time and space required by the data structure to answer boolean queries. Our main contribution is a unified algorithm that recovers the best known results for several boolean queries of practical interests. We then apply our main result to improve upon the state-of-the-art algorithms to answer boolean queries over the four path query which is subsequently used to improve the tradeoff for all path queries of length greater than four and show conditional lower bounds. There are several questions that remain open. We describe the problems that are particularly engaging.

\smallskip
\noindent \textbf{Unconditional lower bounds.} It remains an open problem to prove unconditional lower bounds on the space requirement for answering boolean star and path queries in the RAM model. For instance, $2$-\setdisj can be answered in constant time by materializing all answers using $\Theta(|D|^2)$ space but there is no lower bound to rule out if this can be achieved using sub-quadratic space.

\smallskip
\noindent \textbf{Improved approximate distance oracles.} It would be interesting to investigate whether our ideas can be applied to existing algorithms for constructing distance oracles to improve their space requirement.~\cite{goldstein2017conditional} conjectured that the $k$-reachability tradeoff is optimal and used it to prove the conditional optimality of distance oracles. We believe our framework can be used to improve upon the bounds for $k$-reachability in conjunction with other techniques used to prove bounds for join query processing in the database theory community.

\bibliographystyle{splncs04}
\bibliography{references}
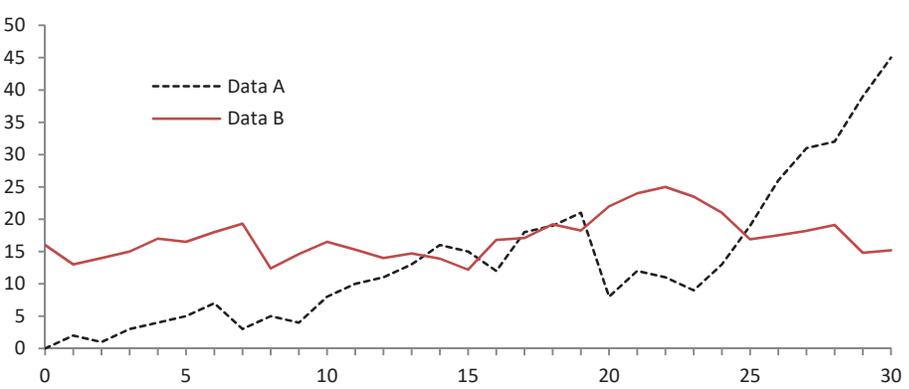
\end{document}